\documentclass[%
 reprint,
 amsmath,amssymb,
 aps,
]{revtex4-1}

\usepackage{algorithm}
\usepackage{algpseudocode}
\usepackage{amsthm}
\usepackage{amsfonts}
\usepackage{bbm}
\usepackage{graphicx}
\usepackage{epstopdf}
\usepackage{dcolumn}
\usepackage{bm}
\usepackage{braket}
\usepackage{color}
\usepackage{enumerate}

\theoremstyle{plain}\newtheorem{theorem}{Theorem}
\theoremstyle{definition}
\theoremstyle{definition}
\theoremstyle{plain}
\theoremstyle{plain}
\theoremstyle{plain}
\theoremstyle{plain}

\begin{document}

\preprint{APS/123-QED}
\title{Distributed quantum algorithm for Simon's problem}

\author{Jiawei Tan$^\dag$, Ligang Xiao, Daowen Qiu$^\ddag$}
\email{$\dag$912574652@qq.com;\\$\ddag$ issqdw@mail.sysu.edu.cn (D.W. Qiu, Corresponding author's address)}
\affiliation{
 Institute of Quantum Computing and Computer Theory, School of Computer Science and Engineering, Sun Yat-sen University, Guangzhou 510006, China \\
 The Guangdong Key Laboratory of Information Security Technology, Sun Yat-sen University, 510006, China\\
QUDOOR Technologies Inc., Guangzhou, China}

\author{Le Luo}
\affiliation{
	School of Physics and Astronomy, Sun Yat-sen University, 519082 Zhuhai, China\\
QUDOOR Technologies Inc., Guangzhou, China
}
\author{Paulo Mateus}
\affiliation{
	Instituto de Telecomunica\c{c}\~{o}es, Departamento de Matem\'{a}tica,
	Instituto Superior T\'{e}cnico,  Av. Rovisco Pais 1049-001  Lisbon, Portugal
}
\date{\today}

\begin{abstract}
Limited by today's physical devices, quantum circuits are usually noisy and difficult to be designed deeply. The novel computing architecture of distributed quantum computing is expected to reduce the noise and depth of quantum circuits. In this paper, we study the Simon's problem in distributed scenarios and design a distributed quantum algorithm to solve the problem.  The algorithm proposed by us has the advantage of exponential acceleration compared with the classical distributed computing, and has the advantage of square acceleration compared with the best distributed quantum algorithm proposed before. In particular, the previous distributed quantum algorithm for Simon's problem can not be extended to the case of more than {\it two computing nodes} (i.e. two subproblems), but our distributed quantum algorithm can be extended to the case of {\it multiple computing nodes} (i.e. multiple subproblems) as well.

\newtheorem{defi}{Definition}

\end{abstract}

\pacs{Valid PACS appear here}
\maketitle


\section{INTRODUCTION}{\label{Sec1}}
Quantum computing has been proved to have great potential in factorizing  large numbers  \cite{shor_polynomial-time_1997}, unordered database search \cite{grover_fast_1996} and chemical molecular simulation \cite{aspuru-guzik_simulated_2005}\cite{cao_quantum_2019}. However, due to the limitations of today's physical devices, large-scale general quantum computers have not been realized. At present, quantum technology has been entered to the NISQ (Noisy Intermediate-scale Quantum) era \cite{preskill_quantum_2018}. So it is possible that we can implement quantum algorithms on middle-scale circuits.

Distributed quantum computing is a novel computing architecture, which combines quantum computing with distributed computing \cite{goos_distributed_2003}\cite{beals_efficient_2013}\cite{Qiu2017DQC}\cite{Qiu22}. In distributed quantum computing architecture, multiple quantum computing nodes are allowed to communicate information through channels and cooperate to complete computing tasks. Compared with centralized quantum computing, the circuit size and depth can be reduced by using distributed quantum computing, which helps to reduce the noise of circuits. In the current NISQ era, the adoption of distributed quantum computing technology is likely conducive to the earlier successful implementation of quantum algorithms.

Simon's problem is an important problem in the history of quantum computing \cite{simon_power_1997}. For Simon's problem, quantum algorithms have the advantage of exponential acceleration over the best classical algorithms \cite{cai_optimal_2018}. Simon's algorithm has a great enlightening effect on the subsequent proposal of Shor's algorithm which can decompose large numbers and compute discrete logarithms in polynomial time \cite{shor_polynomial-time_1997}. Avron et al. proposed a distributed quantum algorithm to solve Simon's problem \cite{avron_quantum_2021}. Their algorithm has the advantage of exponential acceleration compared with the classical algorithm, but in the worst case, it needs $O(n ^ 2)$ queries to solve the Simon's problem (here $n$ represents the length of the string inputted into the oracle in Simon's problem).

In this paper, we study the Simon's problem in distributed scenarios and design a distributed quantum algorithm with query complexity $O(n)$ to solve Simon's problem.  The algorithm proposed by us has the advantage of exponential acceleration compared with the classical distributed computing, and has the advantage of square acceleration compared with the best distributed quantum algorithm proposed before \cite{avron_quantum_2021}. In particular, the previous distributed quantum algorithm can not be extended to the case of more than two computing nodes (i.e. two subproblems), but our distributed quantum algorithm can also be extended to the case of multiple computing nodes (i.e. multiple subproblems), which is important in distributed quantum computing.

The remainder of the paper is organized as follows. In Sec. \ref{Sec2}, we recall Simon's problem and quantum teleportation. Then in Sec. \ref{Sec3}, we introduce the Simon's problem in distributed scenarios by dividing  it into multiple subproblems. After that, in Sec. \ref{Sec4}, we describe the distributed quantum algorithm which we design. The correctness of the algorithm we design is proved in Sec. \ref{Sec5}. In addition,  we compare the efficiency and scalability of our algorithm with distributed classical algorithm and another distributed quantum algorithm proposed in \cite{avron_quantum_2021}. Finally in Sec. \ref{Sec7}, we conclude with a summary. 

\section{PRELIMINARIES}\label{Sec2}

In this section, we serve to review Simon's problem and quantum teleportation that are useful in the paper.


\subsection{Simon's problem}

Simon's problem is a special kind of hidden subgroup problem \cite{kaye_introduction_2007}. We can describe the Simon problem as follows: Consider a function $f:\{0,1\}^n \rightarrow \{0,1\}^m$, where we have the promise that there is a string $s \in \{0,1\}^n$, such that $f(x) = f(y)$ if and only if $x = y$ or $x \oplus y = s$. The $\oplus$ here stands for binary bitwise exclusive or. We have an oracle that can query the value of function $f$. In classical computing, for any $x \in \{0,1\}^n$ and $y \in \{0,1\}^m$, if we input $(x, y)$ into the oracle, we will get $(x, y \oplus f(x))$. In quantum computing, for any $x \in \{0,1\}^n$ and $y \in \{0,1\}^m$, if we input $|x\rangle|y\rangle$ into the oracle, we will get $|x\rangle|y \oplus f(x)\rangle$. Our goal is to find the hidden string $s$ by performing the minimum number of queries to $f$.

On a classical computer, the best algorithm for solving Simon's problem requires $\Theta(\sqrt{2^n})$ queries \cite{cai_optimal_2018}. However, the best quantum algorithm to solve Simon's problem requires $O(n)$ queries \cite{simon_power_1997}. In the history of quantum computing, Simon's problem is the first to show that quantum algorithm has exponential acceleration compared with classical probabilistic algorithm.









\subsection{Quantum teleportation}\label{quantum_teleportation}

Quantum teleportation is an amazing discovery \cite{bennett_teleporting_1993}. By sharing a classical information channel and a pair of entangled states, one can teleport an unknown quantum state to another distant location without actually transmitting physical qubits. This process protects the qubits from being destroyed during transport.

In distributed quantum computing, quantum gates that across multiple computing nodes may be used. We can teleport a quantum state from one computing node to another \cite{caleffi_quantum_2018}, and then apply a multi-qubit gate on the combined state that teleportation can be carried out. 
In this way, we can  implement qubit gates across multiple computing nodes.

\section{Simon's problem in the distributed scenario }\label{Sec3}

In order to better compare our algorithm with the classical distributed algorithm, we describe Simon's problem in the distributed scenario in this section.

In the distributed case, the original function $f$ is divided into two parts, which can be accessed by two subfunctions of domain $\{0,1\}^{n-1}$: $f_e$ and $f_o$($\forall u \in \{0,1\}^{n-1}, f_e(u)=f(u0), f_o(u)=f(u1)$). Consider this scenario: Alice has an oracle $O_{f_e}$ that can query all $f_e(u)$ for all $u \in \{0,1\}^{n-1}$, and Bob has an oracle $O_{f_o}$ that can query all $f_o(u)$ for all $u \in \{0,1\}^{n-1}$.  Here $u0$ (or $u1$)  represents the connection between string $u$ and character $0$ (or $1$). So Alice and Bob's oracles split the domain of function $f$ into two parts at the last bit of the domain. Alice and Bob each know half of the information about the function $f$, but neither knows all of the information about $f$. They need to find the hidden string $s$ by querying their own oracle as few times as possible and exchanging information. The method in  \cite{avron_quantum_2021} solves this problem with $O(n^2)$ queries.

 In this paper, we propose a new distributed quantum algorithm to solve this problem with $O(n)$ queries. In particular, our method can deal with more general case, that is, if there are $2^t$ people, each of whom has an oracle $O_{f_w}$($w \in \{0,1\}^t$ is each person's unique identifier) that can query all $f_w(u)=f(uw)$ for all $u \in \{0,1\}^{n-t}$ ($uw$ here represents the connection between string $u$ and string $w$). So each person can access $2^{n-t}$ values of $f$.
They need to find the hidden string $s$ by querying their own oracle as few times as possible and exchanging information. Actually, this problem cannot be solved by the method in \cite{avron_quantum_2021}.


Next, we further introduce some notations that will be used in the this paper.

\begin{defi}
  For all $u \in \{0,1\}^{n-t}$, let $S(u)$ represent a string of length $2^tm$ by concatenating all strings $f(uw)$ (where $w\in\{0,1\}^t$) according to lexicographical order, that is,
  \begin{align}
  S(u)=f(uw_0)f(uw_1)\cdots f(uw_{2^t-1})
  \end{align}
  where $f(uw_0), f(uw_1),\ldots,f(uw_{2^t-1})\in \{0,1\}^m$ are lexicographical order and $w_i\in\{0,1\}^t$ ($i=0,1,\ldots,2^t-1$) where $w_i\neq w_j$ for any $i\neq j$.
\end{defi}

An example of $S(u)$ is given in Appendix \ref{example}.

\begin{defi}
  For any binary string $u,v \in \{0,1\}^{n}$, $u \cdot v$ denotes $(\sum_{i = 1}^{n}u_i \cdot v_i) \mod 2$.
\end{defi}

\begin{defi}
  For any binary string $u \in \{0,1\}^{n}$, $u^{\perp}$ denotes $\{v \in \{0,1\}^{n}| u \cdot v = 0\}$.
\end{defi}

Let $s$ be the target string to be found, and
 denote $s=s_1s_2$ where the length of $s_1$ is $n-t$. 
 We can query each oracle $O_{f_{w}}$ once to get $f(0^{n-t}w)$ ($w \in \{0,1\}^t$; here for any character $c \in \{0,1\}$ and positive integer $n$, $c^n$ is a binary string of length $n$ which represents $\underbrace{cc\ldots c}_{n}$), and then find $w_1, w_2 \in \{0,1\}^t$ with  $w_1\neq w_2$ such that $f(0^{n-t}w_1)=f(0^{n-t}w_2)$. According to the definition of Simon's problem, we have $0^{n-t}w_1 \oplus 0^{n-t}w_2 = 0^n$ or $0^{n-t}w_1 \oplus 0^{n-t}w_2 = s$. Then we have $w_1 \oplus w_2 = 0^t$ or $w_1 \oplus w_2 = s_2$. Since $w_1 \neq w_2$, we have $w_1 \oplus w_2 = s_2$. So we can obtain $s_2=w_1 \oplus w_2$. If we can not find $w_1$ and $w_2$ such that $f(0^{n-t}w_1)$ is equal to $f(0^{n-t}w_2)$, we can obtain $s_2 = 0^t$.

So the difficulty of Simon's problem in distributed scenario is to find $s_1$ with as few queries as possible. We will introduce our algorithm to find $s_1$ in the next section. It should be noted that our algorithm also works for the original Simon's problem. The purpose in this section is to better compare our algorithm with other classical or quantum distributed algorithms and to illustrate the cases that our distributed quantum algorithm can handle while original Simon's algorithm cannot.

The following theorem concerning $S(u)$ is useful and important.

\begin{theorem}\label{The1} Suppose function $f:\{0,1\}^n \rightarrow \{0,1\}^m$, satisfies that there is a string $s \in \{0,1\}^n$  with $s\neq 0^n$, such that $f(x) = f(y)$ if and only if $x = y$ or $x \oplus y = s$. Then
  $\forall u,v \in \{0,1\}^{n-t},S(u)=S(v)$ if and only if $u \oplus v = 0^{n-t}$ or $u \oplus v = s_1$, where $s=s_1s_2$.
\end{theorem}
\begin{proof}
  Based on the properties of multiset, $\forall u,v\in\{0,1\}^{n-t},S(u)=S(v)$ if and only if $G(u) = G(v)$. So our goal is to prove $\forall u,v \in \{0,1\}^{n-t},G(u)=G(v)$ if and only if $u \oplus v = 0^{n-t}$ or $u \oplus v = s_1$.

  (1) Proof of necessity.

    \qquad (i) If $u \oplus v = 0^{n-t}$, we clearly have $G(u) = G(v)$.

    \qquad (ii) $u \oplus v = s_1$. Since at (i) we have $u \oplus v = 0^{n-t}$, we can now assume that $s_1$ is not $0^{n-t}$. Then we can prove $\forall u \in \{0,1\}^{n-t}, \forall w_1 \neq w_2 \in \{0,1\}^t,f(uw_1)\neq f(uw_2)$. Assume there exists two $t$-bit string $w_1$ and $w_2$ with $w_1 \neq w_2$ and $f(uw_1)= f(uw_2)$, we have $uw_1 \oplus uw_2 = 0^n$ or $uw_1 \oplus uw_2 = s$. Since $w_1 \neq w_2$, we have $uw_1 \oplus uw_2 = s$. Then we have $s = uw_1 \oplus uw_2 = 0^t(w_1 \oplus w_2) = s_1s_2$. This is contrary to $s_1 \neq 0^{n-t}$. So there are no repeating elements in each $G(u)$.

    \qquad \quad Then we can prove $G(u) \subseteq G(v)$. We have $\forall z \in G(u), \exists w \in \{0,1\}^t$ such that $z = f(uw)$. According to the definition of Simon's problem, $f(uw \oplus s)=f(uw)=z$. Then we have $z=f(uw \oplus s)=f((u \oplus s_1)(w \oplus s_2))=f(v(w \oplus s_2)) \in G(v)$. So we have $G(u) \subseteq G(v)$. Similarly we can prove that $G(v) \subseteq G(u)$. As a result, we have $G(u) = G(v)$.

  (2) Proof of adequacy.

  \qquad Since $G(u) = G(v)$, we have $\forall z \in G(u), z \in G(v) $. Then we have $\exists w, w' \in \{0,1\}^t$ such that $z = f(uw)$ and $z = f(vw')$.  As a result, we have $f(uw) = f(vw')$. According to the definition of Simon's problem, we have $uw \oplus vw' = 0^n$ or $uw \oplus vw' = s$.  Therefore, we have $u \oplus v = 0^{n-t}$ or $u \oplus v = s_1$.

\end{proof}

\begin{figure*}[hbtp]
	\begin{minipage}{\linewidth}
		\label{algorithm1}
		\begin{algorithm}[H]
			\caption{Distributed quantum algorithm for Simon's problem (two distributed computing nodes)}
			\begin{algorithmic}[1]
				\State $|\psi_0\rangle = |0^{n-1}\rangle|0^{m}\rangle|0^{m}\rangle|0^{2m}\rangle$;
				
				\State $|\psi_1\rangle = (H^{\otimes n-1}\otimes I^{\otimes 4m})|\psi_0\rangle =(H^{\otimes n-1}|0^{n-1}\rangle)|0^{m}\rangle|0^{m}\rangle|0^{2m}\rangle=\frac{1}{\sqrt{2^{n-1}}}\sum_{u\in\{0,1\}^{n-1}}|u\rangle|0^{m}\rangle|0^{m}\rangle|0^{2m}\rangle$;
				
				\State Each computing node queries its own Oracle under the control of the first quantum register: $|\psi_2\rangle=\frac{1}{\sqrt{2^{n-1}}}|u\rangle|f_e(u)\rangle|f_o(u)\rangle|0^{2m}\rangle=\frac{1}{\sqrt{2^{n-1}}}|u\rangle|f(u0)\rangle|f(u1)\rangle|0^{2m}\rangle$;
				
				\State The fourth quantum register performs its own $U_{Sort}$ under the control of the second and third quantum registers: $|\psi_3\rangle=\frac{1}{\sqrt{2^{n-1}}}|u\rangle|f(u0)\rangle|f(u1)\rangle|(\min(f(u0),f(u1))\max(f(u0),f(u1)))\rangle$;
				
				\State Each computing node queries its own Oracle under the control of the first quantum register: $|\psi_4\rangle=\frac{1}{\sqrt{2^{n-1}}}|u\rangle|0^m\rangle|0^m\rangle|(\min(f(u0),f(u1))\max(f(u0),f(u1)))\rangle$;
				
				\State $|\psi_5\rangle=(H^{\otimes n-1}\otimes I^{\otimes 4m})|\psi_4\rangle$;
				
				\State Measure the first quantum register and get an element in $s_1^{\perp}$.
			\end{algorithmic}
		\end{algorithm}
		
				
			
		\end{minipage}
	\end{figure*}
	
	\begin{figure*}[hbtp]
		\centering
		\includegraphics[width=6in]{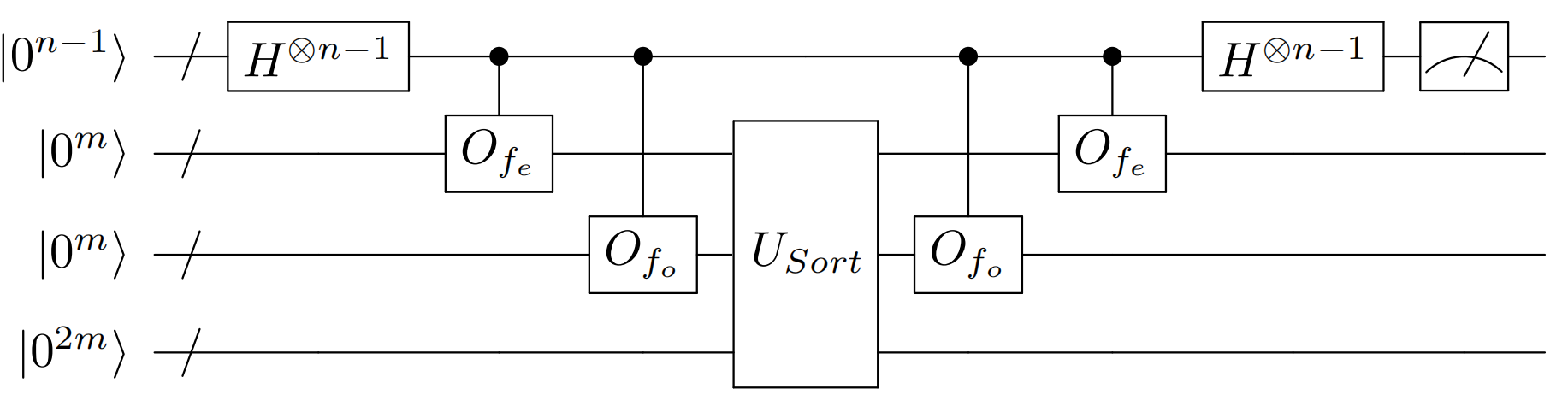}
		\caption{The circuit for the quantum part of distributed quantum algorithm for Simon's problem ($2$ computing nodes).}
		\label{algorithm_2nodes}
	\end{figure*}

\section{Distributed quantum algorithm for Simon's problem}\label{Sec4}

We first give the algorithm with only two distributed computing nodes, i.e. $t=1$. We use an operator $U_{Sort}$ in our algorithm. The effect of operator $U_{Sort}$ in the FIG. \ref{algorithm_2nodes} is: $\forall a,b \in \{0,1\}^{m}$ and $c \in \{0,1\}^{2m}$,
\begin{equation}
	U_{Sort}|a\rangle|b\rangle|c\rangle=|a\rangle|b\rangle|c \oplus (\min(a,b)\max(a,b))\rangle.
\end{equation}

Intuitively, the effect of $U_{Sort}$ is to sort the values in the first two registers, and xor to the third register. $U_{Sort}$ does not change the states of the first two registers (i.e. $|a\rangle$ and $|b\rangle$). We call these two registers the control registers of $U_{Sort}$ . In order to show the control registers of $U_{Sort}$ in quantum circuit diagram clearer, we use $\cdot$ to mark the control registers of $U_{Sort}$ in FIG. \ref{algorithm_2nodes}. Similarly, we use $\cdot$ to mark the control registers of every oracle (i.e. $O_{f_e}/O_{f_o}/O_{f_w}$) in FIG. \ref{algorithm_2nodes} and FIG. \ref{algorithm_multiple_nodes}.


We give a quantum circuit diagram corresponding to our algorithm \cite{nielsen_quantum_2010} as FIG. \ref{algorithm_2nodes}, where  we can implement controlled quantum gates spanning two computing nodes by  using the quantum teleportation described in subsection \ref{quantum_teleportation},


Our algorithm can be further extended to the case of $2^t$ computing nodes. The second and third registers in the Algorithm 1 will be replaced by $2^t$ registers. The $2^t$ registers in the middle will correspond to $2^t$ distributed computing nodes. The extended algorithm is shown in Algorithm 2.

\begin{figure*}[hbtp]
  \begin{minipage}{\linewidth}
	\label{algorithm2}
    \begin{algorithm}[H]
      \caption{Distributed quantum algorithm for Simon's problem ($2^t$ distributed computing nodes)}
      \begin{algorithmic}[1]
        \State $|\phi_0\rangle = |0^{n-t}\rangle(\bigotimes_{w \in \{0,1\}^t}|0^{m}\rangle)|0^{2^tm}\rangle$;

        \State $|\phi_1\rangle = (H^{\otimes n-t}\otimes I^{\otimes 2^{t+1}m})|\phi_0\rangle =(H^{\otimes n-t}|0^{n-t}\rangle)(\bigotimes_{w \in \{0,1\}^t}|0^{m}\rangle)|0^{2^tm}\rangle=\frac{1}{\sqrt{2^{n-t}}}\sum_{u\in\{0,1\}^{n-t}}(\bigotimes_{w \in \{0,1\}^t}|0^{m}\rangle)|0^{2^tm}\rangle$;

        \State Each computing node queries its own Oracle under the control of the first quantum register: $|\phi_3\rangle=\frac{1}{\sqrt{2^{n-t}}}\sum_{u\in\{0,1\}^{n-t}}(\bigotimes_{w \in \{0,1\}^t}|f_w(u)\rangle)|0^{2^tm}\rangle=\frac{1}{\sqrt{2^{n-t}}}\sum_{u\in\{0,1\}^{n-t}}(\bigotimes_{w \in \{0,1\}^t}|f(uw)\rangle)|0^{2^tm}\rangle$;

        \State The fourth quantum register performs its own $U_{Sort}$ under the control of the middle $2^t$ quantum registers: $|\phi_4\rangle=\frac{1}{\sqrt{2^{n-t}}}\sum_{u\in\{0,1\}^{n-t}}(\bigotimes_{w \in \{0,1\}^t}|f(uw)\rangle)|S(u)\rangle$;

        \State Each computing node queries its own Oracle under the control of the first quantum register: $|\phi_6\rangle=\frac{1}{\sqrt{2^{n-t}}}\sum_{u\in\{0,1\}^{n-t}}(\bigotimes_{w \in \{0,1\}^t}|0^{m}\rangle)|S(u)\rangle$;

        \State $|\phi_7\rangle=(H^{\otimes n-t}\otimes I^{\otimes 2^{t+1}m})|\phi_6\rangle$;

        \State Measure the first quantum register and get an element in $s_1^{\perp}$.
      \end{algorithmic}
    \end{algorithm}



  \end{minipage}
  \end{figure*}

Similarly, each computing node uses its oracle twice. The quantum circuit diagram of the extended algorithm is as FIG. \ref{algorithm_multiple_nodes}.

\begin{figure*}[hbtp]
  \centering
  \includegraphics[width=6in]{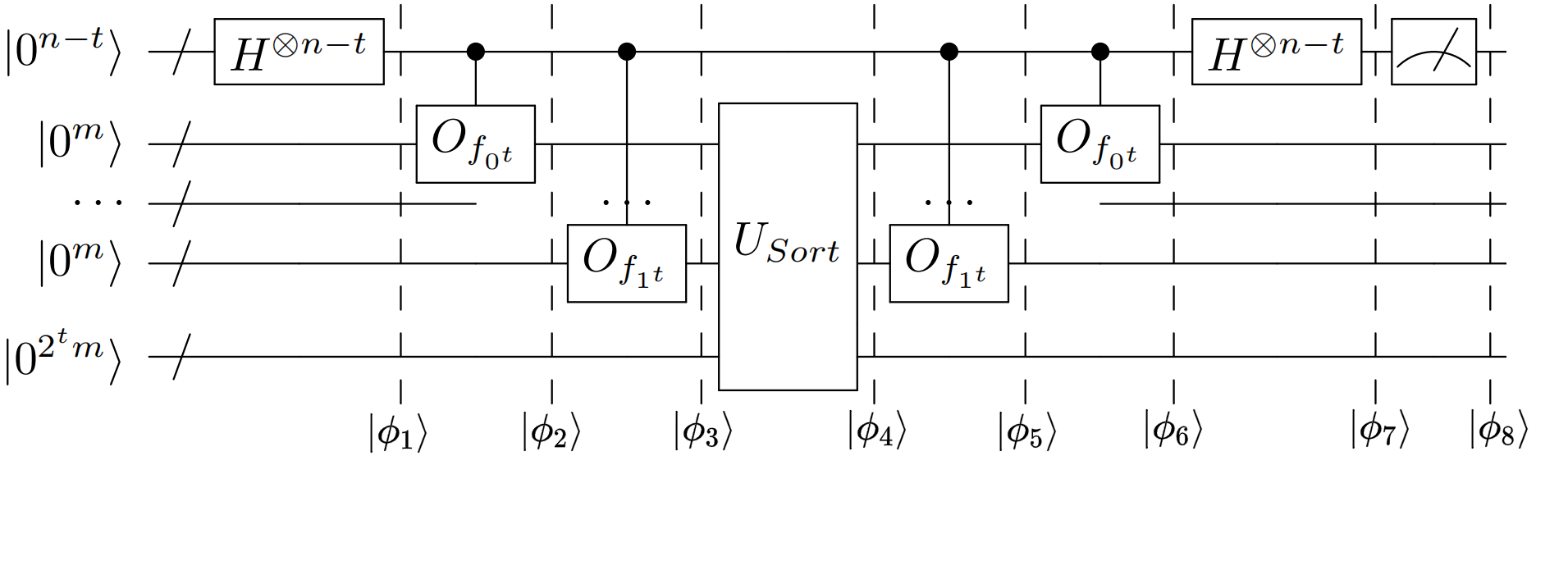}
  \caption{The circuit for the quantum part of distributed quantum algorithm for Simon's problem  ($2^t$ computing nodes).}
  \label{algorithm_multiple_nodes}
\end{figure*}

Similarly, the effect of $U_{Sort}$ in Algorithm 2 is to sort the values in the $2^t$ control registers by lexicographical order, and xor to the target register. It is not difficult to be implemented for sorting multiple elements. By virtue of using sorting network in \cite{paterson_improved_1990}, $2^t$ elements can be sorted in $O(t)$ depth of comparators. Here comparator is a basic circuit module which is easy to be realized.

Note that the oracle query of each quantum computing node in the above algorithm can actually be completed in parallel. In fact, after the first Hadamard transformation,  we can  teleport each group of $n-t$ control bits to every quantum computing node and use this to control the oracle of the computing node.


\section{Correctness analysis of algorithms}\label{Sec5}

In this section, we prove the correctness of our algorithm. First, we define a notation $G(u)$ as follows.

\begin{defi}
  For all $u \in \{0,1\}^{n-t}$, let multiset $G(u) = \{f(uw)|w \in \{0,1\}^t\}$.
\end{defi}
Notice that there could be multiple identical elements in $G(u)$. An example of $G(u)$ is shown in Appendix \ref{example}.

Then, we can write out the state after the first step of the algorithm in FIG. \ref{algorithm_multiple_nodes}.

\begin{align}
  |\phi_1\rangle&=\frac{1}{\sqrt{2^{n-t}}}\sum_{u\in\{0,1\}^{n-t}}|u\rangle(\bigotimes_{w \in \{0,1\}^t}|0^{m}\rangle)|0^{2^tm}\rangle\\
  &=\frac{1}{\sqrt{2^{n-t}}}\sum_{u\in\{0,1\}^{n-t}}|u\rangle\underbrace{|0^m\rangle \ldots|0^m\rangle}_{2^t}|0^{2^tm}\rangle.
\end{align}

The second step of the algorithm queries the oracle $O_{f_{0^t}}$, resulting in the following state:

\begin{align}
  |\phi_2\rangle&=\left(O_{f_{0^t}}\frac{1}{\sqrt{2^{n-t}}}\sum_{u\in\{0,1\}^{n-t}}|u\rangle|0^m\rangle\right)\underbrace{|0^m\rangle \ldots|0^m\rangle}_{2^t-1}|0^{2^tm}\rangle\\
  &=\frac{1}{\sqrt{2^{n-t}}}\sum_{u\in\{0,1\}^{n-t}}|u\rangle|f_{0^t}(u)\rangle \underbrace{|0^m\rangle \ldots|0^m\rangle}_{2^t-1}|0^{2^tm}\rangle\\
  &=\frac{1}{\sqrt{2^{n-t}}}\sum_{u\in\{0,1\}^{n-t}}|u\rangle|f(u0^t)\rangle \underbrace{|0^m\rangle \ldots|0^m\rangle}_{2^t-1}|0^{2^tm}\rangle.
\end{align}

The algorithm then queries each of the other oracles based on the circuit diagram FIG. \ref{algorithm_multiple_nodes} to get the following states:

\begin{equation}
  |\phi_3\rangle=\frac{1}{\sqrt{2^{n-t}}}\sum_{u\in\{0,1\}^{n-t}}|u\rangle\underbrace{|f(u0^t)\rangle \ldots|f(u1^t)\rangle}_{2^t}|0^{2^tm}\rangle.
\end{equation}

After sorting by using $U_{Sort}$, we have the following state:

\begin{equation}
  |\phi_4\rangle=\frac{1}{\sqrt{2^{n-t}}}\sum_{u\in\{0,1\}^{n-t}}|u\rangle\underbrace{|f(u0^t)\rangle \ldots|f(u1^t)\rangle}_{2^t}|S(u)\rangle.
\end{equation}

After that, we query each oracle again and restore the status of the $2^t$ $m$-bit registers to $|0^m\rangle$. Then we obtain the following state:

\begin{equation}
  |\phi_6\rangle=\frac{1}{\sqrt{2^{n-t}}}\sum_{u\in\{0,1\}^{n-t}}|u\rangle\underbrace{|0^m\rangle \ldots|0^m\rangle}_{2^t}|S(u)\rangle.
\end{equation}

By tracing out the states of the $2^t$ $m$-bit registers in the middle, we can get the state:

\begin{equation}
  |\phi'_6\rangle=\frac{1}{\sqrt{2^{n-t}}}\sum_{u\in\{0,1\}^{n-t}}|u\rangle|S(u)\rangle.
\end{equation}

From Theorem \ref{The1}, we know that the structure of the original Simon's problem exists in function $S$. After Hadamard transformation and measurement on the first register, we can get a string that is in $s_1^{\perp}$. After $O(n-t)$ repetitions of the above algorithm, we can obtain $O(n-t)$ elements in the $s_1^{\perp}$. Then, using the classical Gaussian elimination method, we can obtain $s_1$.


We compare our results with other distributed algorithms. Based on the previous analysis, our distributed quantum algorithm needs $O(n-t)$ queries for each oracle $O_{f_u}$ to solve Simon's problem. However, in order to find $s_1$, distributed classical algorithms need to query oracles $O(\sqrt{2^{n-t}})$ times. Our algorithm has the advantage of exponential acceleration compared with the classical distributed algorithm.

For $t=1$, the algorithm in paper \cite{avron_quantum_2021} requires $O(n^2)$ queries to solve Simon's problem. However, our algorithm only needs $O(n)$ queries and has the advantage of square acceleration. In addition, the method  in \cite{avron_quantum_2021} cannot deal with the case of $t > 1$. Therefore, our distributed quantum algorithm has higher scalability.

\section{Conclusion}\label{Sec7}

In this paper, we have designed a distributed quantum algorithm to solve Simon's problem. With multiple quantum computing nodes processing in parallel, each node needs to query own oracle fewer times. This reduces the depth of query complexity for each node. This helps reducing circuit noise and makes it easier to be implemented with exponential acceleration advantages in the current NISQ era.

Our distributed quantum algorithm has the advantage of exponential acceleration compared with the classical distributed algorithm. Compared with previous distributed quantum algorithms, our algorithm can achieve the advantage of square acceleration and has higher scalability.

\section*{Acknowledgements}
This work is supported in part by the National Natural Science Foundation of China (Nos. 61876195, 61572532), and the Natural Science Foundation of Guangdong Province of China (No. 2017B030311011).

\appendix
\section{An example of $G(u)$ and $S(u)$}
\label{example}

We now give a specific function $f$ in Simon's problem and discuss $G(u)$ and $S(u)$ for $t=2$.

Consider $s=1001$ ($s_1=10$ and $s_2=01$) and the function $f:\{0,1\}^4 \rightarrow \{0,1\}^6$ as follow:

\begin{table}[h] 
	\centering
	\caption{An example of function $f$}
	\begin{tabular}{|*{2}{c}|*{2}{c}|*{2}{c}|*{2}{c}|}
		\toprule
		                 $x$       &  $f(x)$   &  $x$    &   $f(x)$ &    $x$       &  $f(x)$   &  $x$    &   $f(x)$  \\
		\hline
		0000          & 100101         & 1000          & 101100     &  0100   &  101010       & 1100          & 011001          \\
		0001     & 101100        & 1001         & 100101    &   0101          &011001          & 1101          & 101010             \\
		0010         & 000100          & 1010 & 110101   & 	0110     & 001101        & 1110         & 111100            \\
		0011   &  110101       & 1011          & 000100  &   	0111         & 111100          & 1111 & 001101          \\
		\toprule
	\end{tabular}
	\label{tab:function_f}
	
\end{table}
We can see that $\forall x,y \in \{0,1\}^4$, $f(x)=f(y)$ if and only if $x \oplus y = 0000$ or $x \oplus y =  1001$ (i.e. $x \oplus y = s$).

Under the above conditions, we have
\begin{align}
	G(00)=&\{f(0000),f(0001),f(0010),f(0011)\}\\
	=&\{100101,101100,000100,110101\},\\
	G(01)=&\{f(0100),f(0101),f(0110),f(0111)\}\\
	=&\{101010,011001,001101,111100\},\\
	G(10)=&\{f(1000),f(1001),f(1010),f(1011)\}\\
	=&\{101100,100101,110101,000100\},\\
	G(11)=&\{f(1100),f(1101),f(1110),f(1111)\}\\
	=&\{011001,101010,111100,001101\}.
\end{align}

So the function $G$ is as follows:
\begin{table}[h] 
	\centering
	\caption{An example of function $G$}
	\begin{tabular}{|*{2}{c}|}
		\toprule
		$u$       &  $G(u)$  \\
		\hline
		00          & \{100101,101100,000100,110101\}                 \\
		01     & \{101010,011001,001101,111100\}         \\
		10         & \{101100,100101,110101,000100\}     \\
		11   &  \{011001,101010,111100,001101\}              \\
		\toprule
	\end{tabular}
	\label{tab:function_G}
\end{table}

We can see that $\forall u,v \in \{0,1\}^2$, $G(u)=G(v)$ if and only if $u \oplus v = 00$ or $u \oplus v =  10$ (i.e. $u \oplus v = s_1$).

For each $G(u)$ ($u \in \{0,1\}^2$), we can obtain $S(u)$ by concatenating all strings in $G(u)$ according to lexicographical order. The function $S$ is as follows:
\begin{table}[h] 
	\centering
	\caption{An example of function $S$}
	\begin{tabular}{|*{2}{c}|}
		\toprule
		$u$       &  $S(u)$  \\
		\hline
		00          & 000100100101101100110101                 \\
		01     & 001101011001101010111100        \\
		10         & 000100100101101100110101    \\
		11   &  001101011001101010111100              \\
		\toprule
	\end{tabular}
	\label{tab:function_S}
\end{table}

We can see that $\forall u,v \in \{0,1\}^2$, $S(u)=S(v)$ if and only if $u \oplus v = 00$ or $u \oplus v =  10$ (i.e. $u \oplus v = s_1$).


\begin{thebibliography}{99}

\bibitem{shor_polynomial-time_1997}P.W. Shor, SIAM J. Comput. 26, 1484 (1997). 

\bibitem{grover_fast_1996}L.K. Grover, in Proceedings of the twenty-eighth annual
ACM symposium on Theory of computing 
(ACM Press, Philadelphia, Pennsylvania, United States,
1996) pp. 212-219.




\bibitem{aspuru-guzik_simulated_2005}
A.~Aspuru-Guzik, A.~D. Dutoi, P.~J. Love, et al., Science 309~(5741) (2005)
1704--1707.



\bibitem{cao_quantum_2019}
Y.~Cao, J.~Romero, J.~P. Olson, et al., Chem. Rev. 119~(19) (2019)
10856--10915.



\bibitem{preskill_quantum_2018}J. Preskill, Quantum 2, 79 (2018).

\bibitem{goos_distributed_2003}H. Buhrman and H. Röhrig, in Mathematical Foundations
of Computer Science 2003, Vol. 2747, edited by G. Goos,
J. Hartmanis, J. van Leeuwen, B. Rovan, and P. Vojtáš
(Springer Berlin Heidelberg, Berlin, Heidelberg, 2003)
pp. 1–20, series Title: Lecture Notes in Computer Science.


\bibitem{beals_efficient_2013}R. Beals, S. Brierley, O. Gray, A.W. Harrow, S. Kutin,
N. Linden, D. Shepherd, and M. Stather, Proc. R. Soc.
A. 469, 20120686 (2013).

\bibitem{Qiu2017DQC}K. Li, D.W. Qiu, L. Li, S. Zheng, and Z. Rong, Information
Processing Letters 120, 23 (2017).

\bibitem{Qiu22} D.W. Qiu, L. Luo, Distributed Grover's algorithm,
arXiv: 2204.10487v3.

\bibitem{simon_power_1997}D. R. Simon, SIAM J. Comput. 26, 1474 (1997).

\bibitem{cai_optimal_2018}G. Cai, D.W. Qiu, Journal of Computer and System
Sciences 97, 83 (2018).

\bibitem{avron_quantum_2021}J. Avron, O. Casper, and I. Rozen, Phys. Rev. A 104,
052404 (2021).



\bibitem{kaye_introduction_2007}P. Kaye, R. Laflamme, and M. Mosca, An introduction
to quantum computing, 1st ed. (Oxford University Press,
Oxford, 2007).

\bibitem{bennett_teleporting_1993}C.H. Bennett, G. Brassard, C. Crépeau, R. Jozsa,
A. Peres, and W.K. Wootters, Phys. Rev. Lett. 70,
1895 (1993).



\bibitem{caleffi_quantum_2018}M. Caleffi, A.S. Cacciapuoti, and G. Bianchi, Proceedings of the 5th ACM International Conference on
Nanoscale Computing and Communication , 1 (2018).

\bibitem{nielsen_quantum_2010}M.A. Nielsen and I.L. Chuang, Quantum computation
and quantum information, 10th ed. (Cambridge University Press, Cambridge ; New York, 2010).

\bibitem{paterson_improved_1990}M.S. Paterson, Algorithmica 5, 75 (1990).
a

\bibitem{inbook-full}D.E. Knuth, in Fundamental Algorithms, The Art of
Computer Programming, Vol. 1 (Addison-Wesley, Reading, Massachusetts, 1973) Section 1.2, pp. 10–119, 2nd
ed., a full INBOOK entry.



\bibitem{feynman_simulating_1982}R.P. Feynman, Int. J. Theor. Phys. 21, 467 (1982).

\bibitem{deutsch_quantum_1985}D. Deutsch, Proc. R. Soc. Lond. A 400, 97 (1985).

\bibitem{deutsch_rapid_1992}D. Deutsch, R. Jozsa, Proc. R. Soc. Lond. A 439,
553 (1992).
\end{thebibliography}

\end{document}